\documentclass[a4paper,10pt]{article} 
\usepackage{graphicx}
 \usepackage{mathptmx}      
\usepackage{latexsym}
\usepackage{amssymb}
\usepackage{amsmath,amsthm}
\usepackage[tmargin=3cm,bmargin=3cm,rmargin=2.5cm,lmargin=2.5cm]{geometry}
\usepackage{dsfont}
\usepackage{subfigure}
\newtheorem{theorem}{Theorem}[section]
\newtheorem{proposition}[theorem]{Proposition}
\newtheorem{lemma}[theorem]{Lemma}
\newtheorem{definition}[theorem]{Definition}

\newtheorem{remark}[theorem]{Remark}

\newcommand{\E}{\ensuremath{\mathbb{E}}}
\newcommand{\N}{\mathbb{N}}
\newcommand{\Z}{\mathbb{Z}}
\newcommand{\R}{\mathbb{R}}
\newcommand{\Q}{\mathbb{Q}}

\title{European Option Pricing of electricity under exponential functional of L\'evy processes with Price-Cap principle 
}

\author{\Large Martin Kegnenlezom \footnote{University of Yaound\'e I,
 Department of Mathematics, P.O. Box 812 Yaound\'e, Cameroon. Phone: (+237)699 299 181, Email: kegnenlezom@gmail.com }, Patrice Takam Soh\footnote{University of Yaound\'e I,
 Department of Mathematics, P.O. Box 812 Yaound\'e, Cameroon. Phone: (+237)699 299 181, Email: patricetakam@gmail.com }, Antoine-Marie Bogso\footnote{University of Yaound\'e I,
 Department of Mathematics, P.O. Box 812 Yaound\'e, Cameroon. Phone: (+237)652 620 452, Email: ambogso@gmail.com }, and \\ \Large Yves Emvudu Wono \footnote{University of Yaound\'e I,
 Department of Mathematics, P.O. Box 812 Yaound\'e, Cameroon. Phone: (+237)699 299 181, Email: yemvudu@yahoo.fr }}
 
\begin{document}
\maketitle

\begin{abstract}
	We propose a new model for electricity pricing based on the price cap principle. The particularity of   the model is that the asset price is an exponential functional of a jump L\'evy process. This model can capture both mean reversion and jumps which are observed in electricity market. It is shown that the value of an European option of this asset is the unique viscosity solution of a partial integro-differential equation (PIDE). A numerical approximation of this solution by the finite differences method is provided. The consistency, stability and convergence results of the scheme are given. Numerical simulations are performed under a smooth initial condition. \\ 
\text{}
\text{}\\
\textbf{subclass MSC: } 60H10, 60H35, 65M06.	

\end{abstract}

\section{Introduction}
\label{intro}
In finance, options are tools that help to guard against risks. But, it is difficult to know the value of an option before the maturity date since, for this end, one has to estimate the value of the underlying in future. In the early of 1970's, Black and Scholes \cite{bs}  brought a major contribution in the evaluation of options. In the case where the underlying is a share that does not pay dividends, they construct a risk-neutral portofolio that replicates the winning profile of an option, which allows to perform the theoretical value of a European option under a closed formula. In the case of the exponential L\'evy Black and Scholes model, this formula is derived from some classical results in discounting, statistics, stochastic and differential calculus. On the contrary, the valuation of options remains an open topic in the case of jump-diffusion processes due to the additional jump term that complicates calculation of option prices. This has been investigated by several authors as  \cite{bates}, \cite{chang} and \cite{scott}. One may distinguish two main approaches used by these authors. The first one, which relies on Fourier transform-based methods, has been introduced by Carr and Madan \cite{carrmadan} to price and analyse European option prices. Many other authors follow the same idea to evaluate option prices. Among them one find \cite{chiarella}, \cite{deng}, \cite{lewis} and \cite{wong}. The advantage of this approach relies on its high computational efficiency when the characteristic function is available. The idea here is to apply the direct discounted expectation method to evaluate the integral of the discounted payoff and risk neutral density function of the underlying process. Continuous and dicrete Fourier transforms were sucessfully applied to pricing options of three exponential L\'evy models: the classical Black-Scoles model (which is a continuous exponential L\'evy model), the Merton jump diffusion model (which is an exponential L\'evy model with finite arrival rate of jumps), and the variance Gamma model (which is an exponential L\'evy model with infinite arrival rate of jumps). However, for many other jump diffusion processes, it is not possible even by applying the L\'evy-Khinchin representation to determine an analytical expression of the characteristic function. This may be due to the complexity of diffusion processes or to the fact that two or more processes are combined in the same model. Therefore one expects to evaluate option prices with rather an approximation of the characteristic function. The debate is then on the choice of a good approximation. Indeed, when the characteristic function can not be expressed explicitly, the Fourier transform method requires two levels of approximation which increases the error level in option valuation. In order to reduce approximation errors, a second approach, namely the Feynman-Kac formula-based approach, has been introduced by some authors such as \cite{alvarez}, \cite{bal} and \cite{eka}. This approach relates the risk-neutral valuation formula to either the solution of a partial differential equation (PDE), when the model is a continuous exponential L\'evy one, or to the solution of a partial integro-differential equation (PIDE) in the case of continuous L\'evy process. The PDE, respectively PIDE obtained may be complex and its theoretical analysis requires new mathematical tools. However, after a numerical approximation step, the solution leads directly to the value of the option. So, we have one level of approximation. Nevertheless, theoretical analysis of the approximation (consistency, stability and convergence of the scheme) remains a challenge. Several finite difference schemes are used in the literature (see e.g. \cite{alvarez}, \cite{ekafinite} and \cite{chap}). Most of these works focus on pricing option when the underlying is driven by either a L\'evy processes or an exponential L\'evy process. Some of the main difficulties are related to the local integral term due to the fact that, on the one hand, the approximation of the risk-neutral density can often involve an infinite summation, and on the other hand, a local integral term requires a specific treatment at both theoretical and numerical levels. We may have other difficulties as the  smoothness of option prices and even the degeneracy of the diffusion coefficient. To overcome such difficulties, the notion of viscosity solution was introduced by Grandall and Lions \cite{lions} for PDEs and, more generally, for PIDEs (see e.g. \cite{alvarez} and \cite{bal}). Precisely, one can split the integro-differential operator into a local and a local part, and then treat the local term using an implicit step, and the local term using an explicit step. This idea was applied by Cont and Voltchkova \cite{ekafinite} to obtain a better approximation of option prices than the previous ones.

In this work, we apply the Feynman-Kac formula-based approach to evaluate option prices in a jump-diffusion model which represents the electricity spot price in a regulated market under the Price-Cap principle. Indeed, in this case, the characteristic function of the jump diffusion process is unavailable. In the next section, we present our jump diffusion model and state the main assumptions. We then taking advantage of the Markov property of the electricity spot price to prove in Section 3 that the European option price solves a  linear PIDE.
In Section 4, we apply an explicit-implicit scheme to compute numerically the viscosity solution of the PIDE provided in the previous section. The consistency, the stability and the convergence of this scheme are studied. The approximation method is similar to \cite{ekafinite} except that contrary to their work (i) the market price is the sum of the exponential L\'evy process and L\'evy process, (ii) the presence of this additional term (see equation \ref{2terms} ) has created an additional difficulty at all levels of the analysis since it still depends on the process, (iii) the parameters used in numerical simulations are derived from the Cameroonian context and (iv) it should also be noted that we worked in the case of time-dependent parameters. Finally, some numerical simulations are performed under a smooth initial condition. 
\section{Electricity Spot Price Model}
We suppose that the underlying asset risk-neutral dynamics is in the form 
\begin{eqnarray}\label{mod}
\mbox{d}S_t = \left(\alpha(t)S_t-\beta(t)\right)\mbox{d}t+ \sigma_t S_t\mbox{d}W_t + (J-1)S_t\mbox{d}q_t,
\end{eqnarray}
with $\alpha(t)=I(t)-EF(t)$ and $\beta(t):=CP(t)+SQ(t)-FU(t)$, where:
\begin{itemize}
	\item[]  $S_t$ represents the electricity spot prices process;
	\item[] $I(t)$ represents the inflation rate; 
	\item [] $EF(t)$ represents the efficiency factor; 
	\item[] $\sigma(t)$ represents the volatility;
	\item[] $W$ represents the standard Brownian motion;
	\item [] $CP(t)$ represents the subvention;
	\item [] $SQ(t)$ represents the service quality penalties, if any;
	\item [] $FU(t)$ represents the uncontrollable cost;
	\item []$J$ represents the proportional random jump size;
	\item []$\mbox{d}q$ a Poisson process such that;
	$\mbox{d}q_t=\left\{\begin{array}{rl}
	0 &\text{with probability}\,\,1- \ell\mbox{d}t\\
	1&\text{with probability} \, \, \ell\mbox{d}t
	\end{array}\right.$.
\end{itemize}
with the following assumptions:
\begin{itemize}
	\item []\textbf{Assumption 1:}  The proportional random jump size $J$ is log-Normally distributed, with $\mathbb{E}\left[J\right]=1$. Hence, $\ln J \sim \mathcal{N}\left(-\frac{\sigma_J^2}{2},\sigma_J^2\right)$. 
	\item []\textbf{Assumption 2:}  The random jump size $J$, the Poisson process $\mbox{d}q_t$ and the diffusion $\mbox{d}W_t$ are independent.
\end{itemize}
This model has been partly inspired from the economic principle price cap proposed by Littlechild in \cite{li} for regulated telecommunication market in UK. This principle has been recently applied to  electricity market \cite{js}. 

We will further assume that the electricity parameters $\alpha(.)$, $\beta(.)$ and $\sigma(.)$ are bounded.
It follows from the It\^{o} formula for jump diffusion process that the exact solution of \eqref{mod} is given by
\begin{equation}\label{2terms}
S_t=S_0e^{X_t}-\int_{0}^{t}\beta(s)e^{X_t-X_s}ds,
\end{equation}
where $$X_t=\int_0^t\alpha(s)-\dfrac{1}{2}\sigma(s)^2ds+\int_{0}^{t}\sigma(s)dW_s+\int_{0}^{t}\ln J dq_s.$$  
\section{Partial Integro-Differential Equation for Option Prices}
This part aims to evaluate the price of an option (Put or Call) in the regulated electricity market under risk-neutral probability, $\Q$, with the terminal payoff, $H_T$, which is given by:
\begin{eqnarray}
C_t= \E[e^{-r(T-t)}H_T|\mathcal{F}_{t}],
\end{eqnarray}
where $r$ represents a free risk discounting rate, $T$ denotes the maturity, and $K$ represents the strike price. Let
$S_T$ be the solution of \eqref{mod} at T, which is the equation of the underlying. 
$H_T=H(S_T)$, with $H(S)=(S-K)^{+}$ for European Call or $H(S)=(K-S)^{+}$ for European Put. From the Markov property, $C_t$ becomes
\begin{equation}\label{po}
C(t,S)= \E[e^{-r(T-t)}H_T|S_{t}=S].
\end{equation}
\begin{proposition}\label{pro1}
	Assume that the European option C given by
	\begin{eqnarray}
	C: (0,T)\times(0,\infty)&\rightarrow&\R \nonumber\\
	(t,S)&\mapsto&C(t,S)
	\end{eqnarray}
	is $C^{1,2}$, with $\partial C/\partial S$ and $\partial ^2C/\partial S^2$ bounded, then $C$ satisfies the partial integro-differential equation:
	\begin{eqnarray}\label{mod1} 
	\dfrac{\partial C}{\partial t}(t,S) + (\alpha(t)S-\beta(t))\dfrac{\partial C}{\partial S}(t,S) + \dfrac{\sigma(t)^2S^2}{2}\dfrac{\partial^2 C}{\partial S^2}(t,S)-rC(t,S)\nonumber\\
	+\ell\int_{\R}\nu(dx)[C(t,xS)-C(t,S)]=0 
	\end{eqnarray}
	on $(0,T)\times(0,\infty)$ with the terminal condition $C(T,S)=H(S)$,  $\;\forall S>0$, where $\ell$ represents the intensity of the Poisson process under risk-neutral measure, and the measure $\nu(dx)$ is the jump size distribution. 
\end{proposition}
\vspace{-0.3cm}
\begin{proof}
	The proof consists of applying It\^{o} formula with jump as in \cite{ali} to the martingale $\tilde{C}(t,S_t)=e^{-rt}C(t,S_t)$. Then the result follows from the fact that the drift term is equal to zero.
	
	By construction $\tilde{C}$ is a martingale. Applying the It\^{o} formula to $\tilde{C}$ we obtain:
	\begin{eqnarray*}
		d\tilde{C_t}&=&e^{-rt}\left[-rC(t,S_t)+\dfrac{\partial C}{\partial t}(t,S_t)+\dfrac{\sigma(t)^2S_t^2}{2}\dfrac{\partial^2 C}{\partial S^2}(t,S_t)\right]dt\nonumber\\
		&&+e^{-rt}\dfrac{\partial C}{\partial S}(t,S_t)dS_t\nonumber\\
		&&+e^{-rt}\left[C(t,JS_{t^-})-C(t,S_{t^-})+(J-1)S_{t^-}\dfrac{\partial C}{\partial S}(t,S_{t^-})\right]dq_t.\nonumber\\
	\end{eqnarray*}
	From \eqref{mod} this equation is equivalent to
	\begin{eqnarray*}
		d\tilde{C_t}&=&e^{-rt}\left[-rC(t,S_t)+\dfrac{\partial C}{\partial t}(t,S_t)+\dfrac{\sigma(t)^2S_t^2}{2}\dfrac{\partial^2 C}{\partial S_t^2}(t,S_t)\right]dt\nonumber\\
		&&+e^{-rt}\left[(\alpha(t)S_t-\beta(t))\dfrac{\partial C}{\partial S}(t,S_t)dt+\dfrac{\partial C}{\partial S}(t,S)S_t\sigma(t)dW_t\right]\\
		&&+e^{-rt}\left[C(t,JS_{t^-})-C(t,S_{t^-})\right]dq_t.\nonumber
	\end{eqnarray*}
	Adding and subtracting 
	\begin{eqnarray*}
		\ell\int_{\R}\nu(dx)(C(t,S_tx)-C(t,S_t))dt,
	\end{eqnarray*}
	one has
	\begin{eqnarray*}
		d\tilde{C_t}=a(t)dt+dM_t, 	
	\end{eqnarray*}
	where
	\begin{eqnarray*}
		a(t)&=&e^{-rt}\left[  \dfrac{\partial C}{\partial t}(t,S_t) + (\alpha(t)S_t-\beta(t))\dfrac{\partial C}{\partial S}(t,S_t) + \dfrac{\sigma(t)^2S_t^2}{2}\dfrac{\partial^2 C}{\partial S^2}(t,S_t)\right. \nonumber\\
		&&- rC(t,S_t)+ \left.\ell\int_{\R}\nu(dx)(C(t,S_tx)-C(t,S_t))\right]
	\end{eqnarray*}
	and
	\begin{eqnarray*}
		dM_t&=&e^{-rt}\left[\dfrac{\partial C}{\partial S}(t,S)S_t\sigma(t)dW_t+
		(C(t,JS_{t^-})-C(t,S_{t^-}))d\tilde{q}_t\right],
	\end{eqnarray*}
	with $\tilde{q}_t=q_t-\ell t$. We now show that $M_t$ is a martingale. Since the payoff function $H$ is Lipschitz. Then, $C$ is also Lipschitz with respect to the second variable $S$. Indeed: 
	\begin{eqnarray*}
		|C(t,x)-C(t,y)|&=& e^{-(T-t)}|\E[H(S_te^{X_T-X_{t}}-\int_{t}^{T}\beta(s)e^{X_T-X_{s}}ds)\mid S_{t}=x]\nonumber\\
		&&-\E[H(S_te^{X_T-X_{t}}-\int_{t}^{T}\beta(s)e^{X_T-X_{s}}ds)\mid S_{t}=y]|\nonumber\\
		&\leq&c_1e^{-r(T-t)}\E[e^{\int_t^T\alpha(s)-(1/2)\sigma(s)^2ds+\int_{t}^{T}\sigma(s)dW_s+\int_{t}^{T}\ln J dq_s}]|x-y|,\; \text{for every fixed $t$}. \nonumber\\
	\end{eqnarray*}
	Since $e^{-\int_t^T(1/2)\sigma(s)^2ds+\int_{t}^{T}\sigma(s)dW_s}$ is a martingale and we also have from assumption 1 that $\E[e^{\int_{t}^{T}\ln J dq_s}]=1$, then we get 
	\begin{equation*}
	|C(t,x)-C(t,y)|\leq c|x-y|e^{\int_t^T\alpha(s)ds}\leq c_1|x-y|, 
	\end{equation*}
	with $c_1=ce^{\int_t^T\alpha(s)ds}$.\\
	Therefore the predictable random function $\varphi(t,x)=C(t,xS_{t^-})-C(t,S_{t^-})$ satisfies:
	\begin{eqnarray*}\label{pre}
		\E \big [\int_{0}^{T}\int_{\R}\nu(dx)|\varphi(t,x)|^2dt\big]&\leq&\E [\int_{0}^{T}dt\int_{\R}\nu(dx)c_1(x^2+1)S_t^2]\nonumber\\
		&\leq&\int_{\R}c_1^2(x^2+1)\nu(dx)\E\big[\int_{0}^{T}S_t^2dt\big]\;<\;\infty,
	\end{eqnarray*}
	where the last inequality holds because the distribution $\nu(dx)$ of the jump sizes is assumed $\log$-normal. Indeed, we have $\int_{\R}x^2\nu(dx)<\infty$, hence $\E[\int_{0}^{T}S_t^2dt]<\infty$. Therefore, the compensated Poisson integral
	$$\int_{0}^{T}\int_{\R}e^{-rt}[C(t,xS_{t^-})-C(t,S_{t^-})]d\tilde{q}_t$$ is a square integrable martingale. Since $C$ is Lipschitz, $\dfrac{\partial C}{\partial S}(t,.)\in L^{\infty}$ and $\big\|\dfrac{\partial C}{\partial S}(t,.)\big\|_{L^{\infty}}\leq c_2$. Thus,
	$$\E \big[\int_{0}^{T}S_t^2|\dfrac{\partial C}{\partial S}(t,S_t)|^2dt\big]\leq c_2^2\E [\int_{0}^{T}S_t^2dt]<\infty.$$ Furthermore, 
	using the isometry formula and the preceding result, it follows that $\int_{0}^{T}\dfrac{\partial C}{\partial S}(t,S_t)S_t\sigma(t)dW_t$ is a square integrable martingale. Therefore, $M_t$ is also a square integrable martingale, implying $\tilde{C}_t-M_t$ is a square integrable martingale. But $\tilde{C}_t-M_t=\int_{0}^{t}a(t)dt$ is also a continuous process with finite variation, so, from Theorem 4.13-450 in \cite{jacod}, one must have $a(t)=0$ $\Q$-almost surely, leading to the PIDE \eqref{mod1}.
\end{proof}
Note that the smoothness (particularly the uniform boundedness of derivatives) assumption made on the European call option is not generally verified as discussed in \cite{eka}. In this case, option prices should be considered as a viscosity solution of the PIDE obtained in Proposition \ref{pro1}. The following proposition gives the link between option prices and the viscosity solution of the PIDE.

\begin{proposition}(\textbf{Option prices as viscosity solutions})\\
	The forward value of the European option defined by \eqref{po} is the (unique) viscosity solution of the Cauchy problem \eqref{mod1}.    
\end{proposition}
\begin{proof}
	Existence and uniqueness of viscosity solutions for such parabolic integro-differential equations are discussed in \cite{alvarez} in the case (the one considered here) where $\nu$ is the finite measure. In what follows, we propose a numerical solution to the PIDE which converges to the viscosity solution as proven in \cite{ekafinite}.
\end{proof}
\section{An Explicit-Implicit Difference scheme} 	  
In this section we present a numerical procedure for solving the PIDE \eqref{mod1} obtained in Proposition \ref{pro1}. Introducing the change of variable $ x=\ln\dfrac{S}{S_0}\quad\text{and}\quad\tau=T-t$ and defining: $u(\tau,x)=e^{r\tau}C(T-\tau,S_0e^{x})$, we obtain: 
\begin{eqnarray}\label{fprice}
u(\tau,x)&=&\E\left[H(S_te^{X_T-X_{T-\tau}}-\int_{T-\tau}^{T}\beta(s)e^{X_T-X_{s}}ds)\mid S_t=S_0e^{x}\right]\nonumber\\
&=& \E\left[H(Y^{x}_{\tau})\right], 
\end{eqnarray}
where $Y^{x}_{\tau}=S_0e^{x+X_T-X_{T-\tau}}-\int_{T-\tau}^{T}\beta(s)e^{X_T-X_{s}}ds $.
We then obtain a PIDE in terms of $u$, given by:
\begin{equation}\label{pid}
\begin{cases}
\dfrac{\partial u}{\partial \tau}=\mathcal{L}u,\,\,on\,\, (0,T]\times O\\
u(0,x)=H(S_0e^{x}),\;\;\;  x\in O,\;\;\;   u(\tau,x)=0,\;\;\;  x\in O^{c},
\end{cases}
\end{equation}

where $O\subset \R$ is an open interval which is not necessarily bounded,  
\begin{eqnarray*} 
\mathcal{L}u(\tau,x)&=&\left(\alpha(T-\tau)-\dfrac{1}{2}\sigma^2(T-\tau)-\dfrac{\beta(T-\tau)}{S_0}\right)\dfrac{\partial u}{\partial x}(\tau,x)+\dfrac{1}{2}\sigma^2(T-\tau)\dfrac{\partial^2 u}{\partial x^2}(\tau,x)\nonumber\\
&&+\ell\int_{\R}\left[u(\tau,x+y)-u(\tau,x)\right]g_{\ln J}(y)dy,
\end{eqnarray*}
with $g_{\ln J}$ denoting the density function of $\ln J$.

The main idea in this method is to split the operator $\mathcal{L}$ into two parts as in \cite{eka}. We replace the differential part with a finite difference approximation, and the integral part with a trapezoidal quadrature approximation. We treat the integral part with an explicit time stepping in order to avoid the inversion problem of the dense matrix $L_J$ associated to the discretization of the integral term.
We then rewrite the PIDE \eqref{pid} as follow:
\begin{equation}\label{pide}
\begin{cases}
\dfrac{\partial u}{\partial \tau}= (\mathcal{L}_D+\mathcal{L}_J)u,\,\,on\,\, (0,T]\times O\\
u(0,x)=H(S_0e^{x}),\;\;\;  x\in O,\;\;\;   u(\tau,x)=0,\;\;\;  x\in O^{c},
\end{cases}
\end{equation}
where
\begin{eqnarray}\label{ca}
\mathcal{L}_Du(\tau,x)&=&\left(\alpha(T-\tau)-\dfrac{1}{2}\sigma^2(T-\tau)-\dfrac{\beta(T-\tau)}{S_0}\right)\dfrac{\partial u}{\partial x}(\tau,x)\nonumber\\&&+\dfrac{1}{2}\sigma^2(T-\tau)\dfrac{\partial^2 u}{\partial x^2}(\tau,x),\\
\mathcal{L}_Ju(\tau,x)&=&\ell\int_{\R}\left[u(\tau,x+y)-u(\tau,x)\right]g_{\ln J}(y)dy.
\end{eqnarray}
Hence, we obtain the approximate problem using the following explicit-implicit time stepping scheme:
\begin{eqnarray}
\dfrac{u^{n+1}-u^{n}}{\Delta t}=L_Du^{n+1}+L_Ju^{n}.
\end{eqnarray}
Before showing the stability of this scheme and applying discretization, the equation must be localized to a bounded domain.
\subsection{Localisation to a bounded domain}
To numerically solve the Cauchy problem \eqref{pide}, we first truncate the space domain to a bounded interval $(-A_l,A_r)$. Usually this leads to defining some boundary conditions at $x=-A_l$ and $x=A_r$. But here, we are in an elliptic local PIDE due to the presence of an integral term. Thus, we need to extend the function $u(\tau,.)$ to a subset $\left\{x+y:x\in (-A_l,A_r),y\in \text{supp}\, g_{\ln J}\right\}$, where $\text{supp}\, g_{\ln J}=\R_+$, is the support of $g_{\ln J}$. Let $u_A(\tau,x)$ be the solution of the following localization problem:
\begin{eqnarray}\label{ca1}
\begin{cases}
\dfrac{\partial u_{l,r}}{\partial \tau}= (\mathcal{L}_D+\mathcal{L}_J)u_{l,r},\,\,on\,\, (0,T]\times (-A_l,A_r)\\
u_{l,r}(0,x)=H(S_0e^{x}),\,\, x\in(-Al,Ar) \,\,\,\,   u_{l,r}(\tau,x)=0 ,\,\, x\notin(-A_l,A_r). 
\end{cases}
\end{eqnarray}
We will show in the following proposition that the localization error decays exponentially with the domain size $A$.
\begin{proposition}\label{pro3}
	Assume $\| H \|_{\infty}<\infty$ and $C_{\tau}=\E 
	\left[e^{\sup\limits_{\eta \in[0,\tau]}|X_T-X_{T-\eta}|}\right]<\infty$. Let $u_{l,r}(\tau,x)$ and $u(\tau,x)$ be respectively the solutions of the Cauchy problems \eqref{pide} and \eqref{ca1}. Then 
	\begin{eqnarray}
	|u(\tau,x)-u_{l,r}(\tau,x)|\leq C_{\tau} \| H \|_{\infty}e^{-\max(A_l,A_r)+|x|},\quad \forall x\in (-A_l,A_r)
	\end{eqnarray}
	where the constant $C_{\tau}$ does not depend on $A_r$ and $A_l$.
\end{proposition} 
\begin{proof}
	Let $M^{x}_{\tau}=\sup\limits_{\eta\in[0,\tau]} |x+X_T-X_{T-\eta}|$. Then 
	\begin{eqnarray}
	&&u_{l,r}(\tau,x)=\E\left[\mathds{1}_{\{M^{x}_{\tau}< \max(Al,Ar)\}}H(Y^{x}_{\tau})\right]\nonumber\\
	\text{and}&&u(\tau,x)=\E\left[H(Y^{x}_{\tau})\right].
	\end{eqnarray}
	Hence
	\begin{eqnarray}\label{d}
	|u-u_{l,r}|&=&\Big|\E 
	\left[H(Y^{x}_{\tau})\mathds{1}_{\{M^{x}_{\tau}\geq \max(A_l,A_r)\}}\right]\Big| \nonumber\\
	&\leq&\|H\|_{\infty}\left|\E\left[\mathds{1}_{\{M^{x}_{\tau}\geq \max(A_l,A_r)\}}\right]\right| \nonumber\\
	&\leq& \|H\|_{\infty}\Q (M^{x}_{\tau}\geqslant\max(Al,Ar).
	\end{eqnarray}
	Theorem 25.18 in \cite{sato} and the fact that $\int_{\R}e^{|x|}\nu(dx)<\infty$ imply
	\begin{eqnarray}
	C_{\tau}=\E\left[e^{\sup\limits_{\eta\in[0,\tau]} |X_T-X_{T-\eta}|}\right]<\infty.
	\end{eqnarray}
	But
	\begin{eqnarray}
	\Q(M^{x}_{\tau}\geqslant \max(Al,Ar))&=&\Q(e^{M^{x}_{\tau}}\geqslant e^{\max(A_l,A_r)})\\
	\end{eqnarray}
	and, since $\sup\limits_{\eta\in[0,\tau]} |x+X_T-X_{T-\eta}|\leq\sup\limits_{\eta\in[0,\tau]} |X_T-X_{T-\eta}|+|x|$, then
	
	$\left\{M^{x}_{\tau}\geqslant \max(Al,Ar)\right\}\subset \left\{\sup\limits_{\eta\in[0,\tau]} |X_T-X_{T-\eta}|
	\geq\max(Al,Ar)-|x|\right\} $.	
	Hence
	\begin{eqnarray}
	\Q \left(e^{M^{x}_{\tau}}\geqslant e^{\max(Al,Ar)}\right)\leq \Q \left(e^{\sup\limits_{\eta\in[0,\tau]} |X_T-X_{T-\eta}|}\geqslant e^{\max(Al,Ar)-|x|}\right).
	\end{eqnarray}
	Now, using Markov's inequality we obtain
	\begin{equation}
	\Q\left(e^{\sup\limits_{\eta\in[0,\tau]} |X_T-X_{T-\eta}|}\geqslant e^{\max(Al,Ar)}\right)
	\leq\dfrac{\E\left[e^{\sup\limits_{\eta\in[0,\tau]} |X_T-X_{T-\eta}|}\right]}{e^{\max(Al,Ar)-|x|}}.
	\end{equation}
	Comparing these last inequalities with (\eqref{d}) gives the desired result.
\end{proof}
\subsection{Truncation of the integral}
To compute numerically the integral term of the PIDE \eqref{pide}, we need to reduce the region of integration to a bounded interval which leads to the truncation of large jumps. We then estimate the error resulting from this approximation. Precisely, suppose a new process, $\tilde{S}_t$, is characterized by the fact that logarithm of the jump size, $\ln \tilde{J}$, is bounded in $[B_l,B_r]$, with the associated measure $\mathds{1}_{\{y\in[Bl,Br]\}}\nu$, where $B_l$ and $B_r$ are real. We further suppose, without loss generality, that $B_l<0$ and $B_r>0$. In this case the corresponding solution to the associated PIDE is denoted by $\tilde{u}(\tau,x)
$. We analyse, in the following proposition, the difference  $|u-\tilde{u}|$.
\begin{proposition}\label{pro4}
	One has:
	\begin{eqnarray}
	|u(\tau,x)-\tilde{u}(\tau,x)|\leq C^{\tau}\left(C_1e^{-|B_l|}+C_2e^{-B_r}\right),
	\end{eqnarray}
	where $C^{\tau}=C\left[\tau \ell S_0+\ell\int_{T-\tau}^{T}\beta(s)(T-s)ds\right]$.
\end{proposition} 
\begin{proof}
	Firstly, let $ \tilde{X_t}$ be a new L\'{e}vy process defined by:
	\begin{equation}\label{modint}
	\tilde{X_t}=\int_0^t\alpha(s)-\dfrac{1}{2}\sigma(s)^2ds+\int_{0}^{t}\sigma(s)dW_s+\int_{0}^{t}\ln \tilde{J} dq_s,
	\end{equation}
	and let:
	\begin{equation}\label{rest}
	\tilde{u}(\tau,x)=\E[H(\tilde{Y}^{x}_{\tau})],
	\end{equation}
	where $\tilde{Y}^{x}_{\tau}=S_0e^{x+\tilde{X}_T-\tilde{X}_{T-\tau}}-\int_{T-\tau}^{T}\beta(s)e^{\tilde{X}_T-\tilde{X}_{s}}ds$.
	Setting $R_{\tau}=X_T-\tilde{X}_T-(X_{\tau}-\tilde{X}_{\tau})$  and using the Lipschitz property on $H$, we obtain: 
	\begin{eqnarray}\label{in}
	|u(\tau,x)-\tilde{u}(\tau,x)|&=&|\E[H(Y^{x}_{\tau})]-\E[H(\tilde{Y}^{x}_{\tau})]|\nonumber\\
	&\leq&c_1\E\left[|S_0(e^{x+\tilde{X}_T-\tilde{X}_{T-\tau}+R_{T-\tau}}-e^{x+\tilde{X}_T-\tilde{X}_{T-\tau}})\right.\nonumber\\
	&&\left.-\int_{T-\tau}^{T}\beta(s)(e^{(\tilde{X}_T-\tilde{X}_{s}+R_{s}}-e^{\tilde{X}_T-\tilde{X}_{s}}ds|\right]\nonumber\\
	&\leq&c_1\left(S_0\E[e^{\tilde{X}_T-\tilde{X}_{T-\tau}}|e^{R_{T-\tau}}-1|]\right.\nonumber\\&&+ \left.\int_{T-\tau}^{T}\beta(s)\E[e^{\tilde{X}_T-\tilde{X}_{s}}|e^{R_{s}}-1|]ds\right).
	\end{eqnarray}
	Since $R_{\tau}$ and $\tilde{X}_{T}-\tilde{X}_{\tau} $ are independent, we have   
	\begin{eqnarray}\label{inn}
	|u(\tau,x)-\tilde{u}(\tau,x)|&\leq& c_1e^{x}(S_0\E[e^{\tilde{X}_T-\tilde{X}_{T-\tau}}]\E[|e^{R_{T-\tau}}-1|]+\int_{T-\tau}^{T}\beta(s)\E[e^{\tilde{X}_T-\tilde{X}_{s}}]\E[|e^{R_{s}}-1|]ds). 
	\end{eqnarray}
	Moreover, $\left(e^{\tilde{X}_T-\tilde{X}_{T-u}}, u\in[0,T]\right)$ being a martingale, $\E[e^{\tilde{X}_T-\tilde{X}_{T-\tau}}]=1$ and $\E[e^{\tilde{X}_T-\tilde{X}_{s}}]=1$. As a consequence,  
	\begin{eqnarray}\label{eq:inn1}
	|u(\tau,x)-\tilde{u}(\tau,x)|&\leq& c_1e^{x}(S_0\E[|e^{R_{T-\tau}}-1|]+\int_{T-\tau}^{T}\beta(s)\E[|e^{R_{s}}-1|]ds). 
	\end{eqnarray}
	Since, for every $a\in \R$, $|e^{a}-1|=(e^{a}-1)+2(e^{a}-1)^+$ and $(e^{a}-1)^+\leq |a|$, then	
	\begin{eqnarray}\label{inn1}
	|u(\tau,x)-\tilde{u}(\tau,x)|&\leq& c_1e^{x}(S_0\E[|R_{T-\tau}|]+\int_{T-\tau}^{T}\beta(s)\E[|R_{s}|]ds). 
	\end{eqnarray}	
	
	But:
	\begin{eqnarray}\label{inegali}
	\E[|R_{T-\tau}|]&\leq& \ell\int_{T-\tau}^{T}\left[-\int_{-\infty}^{Bl}yg_{\ln J}(y)dy+\int_{Br}^{+\infty}yg_{\ln J}(y)dy\right]ds\nonumber\\
	&\leq&\tau \ell \left(-e^{-|Bl|}\int_{-\infty}^{Bl}ye^{|y|}g_{\ln J}(y)dy+e^{-|Br|}\int_{Br}^{+\infty}ye^{|y|}g_{\ln J}(y)dy\right)\nonumber\\ 
	&\leq&\tau \ell S_0\left(C_1e^{-|Bl|}+C_2e^{-|Br|}\right),
	\end{eqnarray}
	where $C_1=-\int_{-\infty}^{Bl}ye^{|y|}g_{\ln J}(y)dy$ and $C_2=\int_{Br}^{+\infty}ye^{|y|}g_{\ln J}(y)dy$. Replacing \eqref{inegali} into \eqref{in}, we get:
	\begin{eqnarray}
	|u(\tau,x)-\tilde{u}(\tau,x)|&\leq&C\left[\tau \ell S_0+\ell\int_{T-\tau}^{T}\beta(s)(T-s)ds\right]\left(C_1e^{-|Bl|}+C_2e^{-|Br|}\right),
	\end{eqnarray}
	where $C=c_1e^{x}$
\end{proof}
From Proposition\ref{pro3} and \ref{pro4}, $\tilde{u}$ converges to $u$ when $|B_l|$ and $|B_r|$ grow to infinity.
\subsection{Explicit-implicit scheme}
Define a uniform grid on $(0,T]\times(-A_l,A_r)$ by $\tau_{n}=n\Delta t$, $n=0$,.....,$M$, $x_i=i\Delta x-A_l$, $i=0,..,N$, with $\Delta t=T/M$ and $\Delta x=\dfrac{A_r+A_l}{N}$. Let $(u^{n}_{i})$ be the solution of the numerical scheme which we define below: Firstly, to approximate the integral terms, we use the trapezoidal quadrature rule with the same resolution $\Delta x$. Let $K_l$ and $K_r$ be such that $[B_l,B_r]\subset[(K_l-1/2)\Delta x,(K_r+1/2)\Delta x]$, then:
\begin{equation}\label{eq}
\int_{B_l}^{B_r}(u(\tau,x_i+y)-u(\tau,x_i))g_{\ln J}(y)dy\simeq \sum\limits_{j=K_l}^{K_r}\nu_j(u_{i+j}-u_i), 
\end{equation}
where $\nu_j=\int_{(j-1/2)\Delta x}^{(j+1/2)\Delta x} g_{\ln J}(y)dy$. Notice that to compute the integral term, we need to extend the solution to $[-A_l+B_l,Ar+B_r]$. Hence, we assume that this solution is zero except in $[-A_l,A_r]$. The derivatives are discretized using the finite difference method thus:
\begin{equation}\label{eqa}
\left\{\begin{array}{rl}
\left(\dfrac{\partial ^{2}u}{\partial x^2}\right)_{i}\simeq& \dfrac{u_{i+1}-2u_{i}+u_{i-1}}{(\Delta x)^2}\\
\left(\dfrac{\partial u}{\partial x}\right)_{i}\simeq& \left\{\begin{array}{rl}\dfrac{u_{i+1}-u_i}{\Delta x}&\mbox{if}\, f(\tau,x)\geq 0\\
\dfrac{u_{i}-u_{i-1}}{\Delta x}&\mbox{if}\, f(\tau,x)\leq 0,
\end{array}\right.
\end{array}\right.\end{equation}
where $f(\tau,x)=\alpha(T-\tau)-\dfrac{1}{2}\sigma^2(T-\tau)-\dfrac{\beta(T-\tau)}{S_0}e^{x}$.\\
Using \eqref{eq} and \eqref{eqa}, and supposing $f(\tau,x)<0$, we obtain the following relation:
\begin{equation}
\dfrac{u^{n+1}_{i}-u^{n}_{i}}{\Delta t}=\left(L_Du\right)^{n+1}_{i}+\left(L_Ju\right)^{n}_{i},
\end{equation}
where
\begin{equation}
\left\{\begin{array}{rl}
\left(L_Du\right)^{n}_{i}=&f(\tau_{n},x_i)\dfrac{u^{n}_{i+1}-u^{n}_i}{\Delta x}+\dfrac{1}{2}\sigma^2(T-\tau_{n})\dfrac{u^{n}_{i+1}-2u^{n}_{i}+u^{n}_{i-1}}{(\Delta x)^2}\\
\left(L_Ju\right)^{n}_{i}=& \sum\limits_{j=K_l}^{K_r}\nu_j(u^{n}_{i+j}-u^{n}_{i}).
\end{array}\right.
\end{equation}
Finally, we replace the problem \eqref{pide} with the following time-stepping numerical scheme:
\begin{equation}\label{fds}
\left\{\begin{array}{ll}
\text{Initialisation}\,\,u^0_i=H(S_0e^{x_i})& \text{if}\,\, i\in \{0,...N-1\}\\
\text{For n=0,...,M-1}&\\
\dfrac{u^{n+1}_{i}-u^{n}_{i}}{(\Delta t)}=\left(L_Du\right)^{n+1}_{i}+\left(L_Ju\right)^{n}_{i}&\text{if}\,\,i\in\{0,..,N-1\}\\
u^{n+1}_{i}=0&\text{if}\,\, \notin\{0,..,N-1\}.
\end{array}\right.
\end{equation}
After defining the numerical scheme, we study some of its important properties, particularly, consistency, monotonicity, stability and convergence.
\subsection{Consistency}
The follow proposition shows that (\ref{fds}) is consistent with \eqref{pide}. 
\begin{proposition}(\textbf{Consistency})\\	
	The finite difference scheme \eqref{fds} is locally consistent with equation \eqref{pide}: That is,  $\forall \, v\in C^{\infty}_{0}([0,T]\times(A_l,A_r))$, and $\forall\,(\tau_n,x_i)\in [0,T]\times\R$, one has:
	\begin{equation}
	\left|\dfrac{v^{n+1}_{i}-v^{n}_{i}}{(\Delta t)}-
	\left(L_Dv\right)^{n+1}_{i}-\left(L_Jv\right)^{n}_{i}-\dfrac{\partial v}{\partial \tau}(\tau_n,x_i)-(\mathcal{L}_D+\mathcal{L}_J)v(\tau_n,x_i)\right|=r^{n}_{i}(\Delta t,\Delta x)\rightarrow 0
	\end{equation}
	as $(\Delta t,\Delta x)\rightarrow (0,0)$. In other words,
	$\exists c>0$ such that:  $\;\;|r^{n}_{i}(\Delta t,\Delta x)|\leq c(\Delta t+\Delta x).$
\end{proposition}
\begin{proof}
	Let
	\begin{equation}\label{coeficient}
	\begin{cases}
	a_1=\dfrac{v^{n+1}_{i}-v^{n}_{i}}{\Delta t}-\dfrac{\partial v}{\partial \tau}(\tau_n,x_i)\\
	a_2=\left(L_Dv\right)^{n+1}_{i}-\mathcal{L}_Dv(\tau_n,x_i)\\
	a_3=\left(L_Jv\right)^{n}_{i}-\mathcal{L}_Jv(\tau_n,x_i).
	\end{cases}
	\end{equation}
	Using the second order Taylor expansion with respect to $\tau$, we obtain
	$$v^{n+1}_{i}\approx v^{n}_{i}+\Delta t\dfrac{\partial v}{\partial \tau}(\tau_n,x_i)+\dfrac{(\Delta t)^{2}}{2}\dfrac{\partial^2 v}{\partial \tau^2}(\tau_n,x_i).$$ 
	Plugging this relation in the first equation in (\eqref{coeficient}) we get:
	\begin{eqnarray}\label{inga1}
	|a_1|= \dfrac{\Delta t}{2}\left|\dfrac{\partial^2v }{\partial \tau^2}(\tau_n,x_i)\right|\leq\dfrac{\Delta t}{2}\left\|\dfrac{\partial^2v}{\partial \tau^2}(\tau_n,x_i)\right\|_{\infty}=\dfrac{\Delta t}{2}M,
	\end{eqnarray}
	where $M=\left\|\dfrac{\partial^2v}{\partial \tau^2}\right\|_{\infty}$.	We now show that $|a_2|$ is bounded. By the mean-value theorem, there exists $\theta\in ]\tau_n,\tau_{n+1}[$\; such that: $$\mathcal{L}_Dv(\tau_{n+1},x_i)-\mathcal{L}_Dv(\tau_n,x_i)\approx \Delta t \partial_{\tau}\mathcal{L}_Dv(\tau_n+\Delta t\theta,x_i).$$
	If one replaces this relation in the second equation in (\eqref{coeficient}), then:
	\begin{equation}\label{a2}
	|a_2|=\left|\left(L_Dv\right)^{n+1}_{i}-\mathcal{L}_Dv(\tau_{n+1},x_i)+\Delta t \partial_{\tau}\mathcal{L}v(\tau_n+\Delta t\theta,x_i)\right|.
	\end{equation}
	Next, taking Taylor expansion of $v$ of order $4$ gives:
	\begin{eqnarray}
	v^{n+1}_{i+1}&\approx& v^{n}_{i}+\Delta x\dfrac{\partial v}{\partial x}(\tau_{n+1},x_i)+\dfrac{(\Delta x)^2}{2}\dfrac{\partial^2 v}{\partial x^2}(\tau_{n+1},x_i);\nonumber\\
	&&+ \dfrac{(\Delta x)^3}{6}\dfrac{\partial^3 v}{\partial x^3}(\tau_{n+1},x_i)+\dfrac{(\Delta x)^4}{24}\dfrac{\partial^4 v}{\partial x^4}(\tau_{n+1},x_i)\nonumber
	\end{eqnarray}
	\begin{eqnarray}
	v^{n+1}_{i-1}&\approx& v^{n+1}_{i}-\Delta x\dfrac{\partial v}{\partial x}(\tau_{n+1},x_i)+\dfrac{(\Delta x)^2}{2}\dfrac{\partial^2 v}{\partial x^2}(\tau_{n+1},x_i)\nonumber\\
	&&- \dfrac{(\Delta x)^3}{6}\dfrac{\partial^3 v}{\partial x^3}(\tau_{n+1},x_i)+\dfrac{(\Delta x)^4}{24}\dfrac{\partial^4 v}{\partial x^4}(\tau_{n+1},x_i),\nonumber
	\end{eqnarray}
	hence
	\begin{eqnarray}
	\dfrac{v^{n+1}_{i+1}-2v^{n+1}_{i}+v^{n+1}_{i-1}}{(\Delta x)^2}\approx
	\left((\Delta x)^2\dfrac{\partial^2v}{\partial x^2}+\dfrac{(\Delta x)^3}{12}\dfrac{\partial^4v}{\partial x^4}\right).\nonumber
	\end{eqnarray}
	Putting this last result in (\eqref{a2}) gives: 
	\begin{eqnarray}
	|a_2|&\leq& \dfrac{(\Delta x)^2}{6}|f(\tau_{n+1},x_i)| \left|\dfrac{\partial^3v}{\partial x^3}+\dfrac{(\Delta x)}{4}\dfrac{\partial^4v}{\partial x^4}\right|+\dfrac{(\Delta x)^2}{6}\dfrac{\sigma^2(T-\tau_{n+1})}{4}\left|\dfrac{\partial^4v}{\partial x^4}\right|\nonumber\\&&+ \Delta t |\partial_{\tau}\mathcal{L}v(\tau_n+\Delta t\theta,x_i)|,
	\end{eqnarray}
	since $\alpha$, $\sigma$, $\beta$ and $f$ are bounded functions. Also, since the derivatives $\partial^{m+n}v\big/\partial\tau^n\partial x^m$ are bounded, it implies:
	\begin{eqnarray}\label{inga2}
	|a_2|\leq (\Delta x)^2M_1+\Delta tM_2.
	\end{eqnarray}
	One also has: 
	\begin{eqnarray*}
		|a_3|&=&\left| \sum\limits_{j=K_l}^{K_r}\nu_j(v^{n}_{i+j}-v^{n}_{i})-\int_{B_l}^{B_r}(v(\tau,x_i+y)-v(\tau,x_i))g_{\ln J}(y)dy\right|\\
		&=&\left| \sum\limits_{j=K_l}^{K_r}\int_{(j-1/2)\Delta x}^{(j+1/2)\Delta x}(v^{n}_{i+j}-v^{n}_{i}) g_J(y)dy
		-\int_{B_l}^{B_r}(v(\tau,x_i+y)
		-v(\tau,x_i))g_J(y)dy\right|.
	\end{eqnarray*}
	Since  $[B_l,B_r]\subset[(K_l-1/2)\Delta x,(K_r+1/2)\Delta x]$, then we have:
	\begin{equation*}
	|a_3|\leq\left| \sum\limits_{j=K_l}^{K_r}\int_{(j-1/2)\Delta x}^{(j+1/2)\Delta x}(v(\tau_n,x_i+y_j)-v(\tau_n,x_i+y) g_{\ln J}(y)dy\right|,
	\end{equation*}
	and using Taylor's expansion of order one, we get:
	\begin{equation*}
	|a_3|\leq 	\left| \sum\limits_{j=K_l}^{K_r}\int_{(j-1/2)\Delta x}^{(j+1/2)\Delta x}(y_j-y)\dfrac{\partial v}{\partial x}(\tau_n,x_i+\psi)g_{\ln J}(y)dy\right|,\,\ \psi \in ]x_i+y,x_i+y_j[.
	\end{equation*}
	From the scheme, we have $\Delta x(j-1/2)\leq y \leq \Delta x(j+1/2)$, which leads to $-\dfrac{\Delta x}{2}\leq y_j-y\leq \dfrac{\Delta x}{2} $. Hence,
	\begin{eqnarray}\label{inga3}
	|a_3|&\leq& \dfrac{\Delta x}{2}\left| \sum\limits_{j=K_l}^{K_r}\int_{(j-1/2)\Delta x}^{(j+1/2)\Delta x}\dfrac{\partial v}{\partial x}(\tau_n,x_i+\psi) g_{\ln J}(y)dy\right|\nonumber\\
	&\leq& \dfrac{\Delta x}{2}\left\|\dfrac{\partial v}{\partial x}\right\|_{\infty} \left| \sum\limits_{j=K_l}^{K_r}\int_{(j-1/2)\Delta x}^{(j+1/2)\Delta x} g_{\ln J}(y)dy\right|=\dfrac{\Delta x}{2}M_3,
	\end{eqnarray}
	where $M_3=\left\|\dfrac{\partial v}{\partial x}\right\|_{\infty} \left| \sum\limits_{j=K_l}^{K_r}\int_{(j-1/2)\Delta x}^{(j+1/2)\Delta x} g_{\ln J}(y)dy\right|$. Finally, \eqref{inga1}, \eqref{inga2} and \eqref{inga3} imply
	$$|r^{n}_{i}(\Delta t,\Delta x)|\leq \Delta t\left(\dfrac{M}{2}+M_1\right)+\Delta x\left(\Delta xM_2+\dfrac{M_3}{2}\right)\rightarrow 0$$ 
	as $(\Delta t, \Delta x)\rightarrow(0,0)$. 
\end{proof}
\subsection{ Stability and monotonicity} 
Two properties are important to show convergence to viscosity solutions: stability and monotonicity of scheme.
\begin{definition}\textbf{Stability}\\
	The scheme \eqref{eqa} is stable if, and only if, for some bounded initial conditions, its solution exists and is bounded independently of $\Delta t$ and $\Delta x$, and uniformly bounded on $[0,T]\times\R$. That is to say,
	$$\exists C>0,\;\; \forall \Delta t>0,\;\; \forall \Delta x>0 ,\;\; i\in \Z,\;\;n\in\{0,...,M\},\;\;|u^{n}_{i}|\leq C.$$
\end{definition}
We will say that a given vector $v$ is positive if all its elements are positive. We write $u\geq v$ if $u-v\geq 0$. In this part we show the stability property of the scheme, which in turn implies the discrete comparison principle, a property which has an important interpretation in finance. This property makes possible the fact that the options values computed using our numerical scheme will check arbitrage inequalities: Inequality between payoffs leading to inequality between options values.

\begin{proposition}(\textbf{Stability and the discrete comparison principle})\label{Stability}\\
	If $\Delta t\leq 1\big/\sum_{j=K_l}^{K_r}\nu_j$, the scheme \eqref{eqa} is stable, and hence verifies the discrete comparison principle:
	$$u^0\geq v^0 \Longrightarrow \forall n\in \N^*,\;\; u^n\geq v^n.$$
\end{proposition}
\begin{proof}
	Firstly, consider \eqref{eqa} in the form:
	\begin{eqnarray}\label{pv1}
	-cu^{n+1}_{i-1}+(1+a\Delta t)u^{n+1}_{i}-b\Delta t
	u^{n+1}_{i+1}=\left(1-\Delta t\sum_{j=K_l}^{K_r}\nu_j\right)u^{n}_{i}+\Delta t\sum_{j=K_l}^{K_r}u^{n}_{i+j}\nu_j,
	\end{eqnarray}
	where
	\begin{equation}
	\left\{\begin{array}{ll}
	c&=\dfrac{1}{2}\dfrac{1}{(\Delta x)^2}\sigma^2(T-\tau_{n+1})\geq 0\\
	a&=\dfrac{1}{\Delta x}f(\tau_{n+1},x_i)+\dfrac{1}{(\Delta x)^2}\sigma^2(T-\tau_{n+1})\geq 0\\
	b&=\dfrac{1}{\Delta x}f(\tau_{n+1},x_i)+\dfrac{1}{2}\dfrac{1}{(\Delta x)^2}\sigma^2(T-\tau_{n+1})\geq 0.
	\end{array}\right.
	\end{equation}
	The positivity of $a$ and $b$ arises from $g$ being positive. If $g$ is not positive, we change the approximation of the first-order derivatives in the scheme used. In either case, one has: 
	$$ a=b+c \;\;\Rightarrow \;\;a\Delta t=b\Delta t+c\Delta t\;\; \Rightarrow \;\; 1+a\Delta t>(b+c)\Delta t .$$
	It follows that $|1+a\Delta t|> |-c\Delta t|+|-b\Delta t|$, implying the matrix of linear system on $(u^{n+1}_{0},...,u^{n+1}_{N}) $ has a strict dominant diagonal, hence invertible. Therefore, the solution of the linear system exists and is unique.	We now show by mathematical induction that this solution is bounded. That is, if $\| H\|_{\infty}\leq \infty$ is the bounded initial condition, then, $ \forall n\in \N$,
	
	\begin{equation}\label{pv}
	\| u^{n}\|_{\infty}\leq\| H\|_{\infty}.  
	\end{equation}
	By definition of $u^0$, we have $\| u^{0}\|_{\infty}\leq\| H\|_{\infty}$. Assume \eqref{pv} holds for $n$. To show that it holds for $n+1$, we suppose on the contrary that $ \| u^{n+1}\|_{\infty}>\| H\|_{\infty} $. By the definition of $\| . \|_{\infty}$,  $\exists i_0\in \{0,...,n\}$ such that $ | u^{n+1}_{i_0}|=\| u^{n+1}\|_{\infty}  $, and $ \forall i\in \Z ,$ $\;| u^{n+1}_{i}|<|u^{n+1}_{i_0}|
	$.\\
	Since $a=b+c$, we can write, 
	\begin{equation}
	\| u^{n+1}\|_{\infty}=|u^{n+1}_{i_0}|
	=-c\Delta t|u^{n+1}_{i_0}|+(1+a\Delta t)|u^{n+1}_{i_0}|-b\Delta t|u^{n+1}_{i_0}|.                                                 
	\end{equation}
	Moreover, as $|u^{n+1}_{i_0-1}|<|u^{n+1}_{i_0}|$ and $|u^{n+1}_{i_0+1}|<|u^{n+1}_{i_0}|$ we have
	
	\begin{equation}\label{inq}
	\| u^{n+1}\|_{\infty}
	\leqslant-c\Delta t|u^{n+1}_{i_0-1}|+(1+a\Delta t)|u^{n+1}_{i_0}|-b\Delta t|u^{n+1}_{i_0+1}|.
	\end{equation}
	Using \eqref{pv1} and \eqref{inq}, and the fact that $\Delta t\leq1\big/\sum_{j=K_l}^{K_r}\nu_j$, we obtain:
	\begin{eqnarray*}
		\| u^{n+1}\|_{\infty}
		&\leq&\left|\left(1-\Delta t\sum_{j=K_l}^{K_r}\nu_j\right)u^{n}_{i_0}+\Delta t\sum_{j=K_l}^{K_r}u^{n}_{i_0+j}\nu_j\right|\nonumber\\
		&\leq&\left(1-\Delta t\sum_{j=K_l}^{K_r}\nu_j\right)|u^{n}_{i_0}|+\Delta t\sum_{j=K_l}^{K_r}|u^{n}_{i_0+j}\nu_j|\nonumber\\
		&\leq&\left(1-\Delta t\sum_{j=K_l}^{K_r}\nu_j\right)\| u^{n}\|_{\infty}+\Delta t\sum_{j=K_l}^{K_r}\nu_j\| u^{n}\|_{\infty}\nonumber\\
		&=&\| u^{n}\|_{\infty}\leq \|H\|_{\infty},
	\end{eqnarray*}                                                   
	which contradicts our assumption. Hence $\|u^{n}\|_{\infty}\leq \| H\|_{\infty}$.
\end{proof} 
\begin{proposition}(\textbf{Monotonicity})\\
	Let $u^n$ and $v^n$ be two solutions to \eqref{eqa} corresponding to some initial conditions $f$ and $h$ respectively, satisfying $f(x)\geq h(x)$ $\;\forall x\in\R$. If $\Delta t\leq1\big/\sum_{j=K_l}^{K_r}\nu_j$, then $u^n\geq v^n$, $\forall n\in \N$.
\end{proposition}
\begin{proof}
	Define $w^n=u^n-v^n$. We show that $w^n\geq 0$ $\;\forall n\in \N$. As in Proposition \ref{Stability}, we proceed by induction. By construction, we have $w^0_i=f(x_i)-b(x_i)\geq 0$, $\forall i\in \Z$. Let $w^n\geq 0$, and suppose that: $\inf\limits_{i\in \Z}w^{n+1}_{i}<0$. Since $\forall i\in \Z\backslash\{0,...,N\}$, $w^{n+1}_{i}=0$, this implies that $\exists i_0\in\{0,...,N\}$ such that $ w^{n+1}_{i_0}=\inf\limits_{i\in \Z}w^{n+1}_{i}$. Using \eqref{pv1} and $\Delta t\leq \dfrac{1}{\sum_{j=K_l}^{K_r}\nu_j}$, we have that
	\begin{eqnarray*}
		\inf\limits_{i\in \Z}w^{n+1}_{i}=w^{n+1}_{i_0}&=&-c\Delta tw^{n+1}_{i_0}+(1+a\Delta t)w^{n+1}_{i_0}-b\Delta tw^{n+1}_{i_0}\\
		&\geq&-c\Delta tw^{n+1}_{i_0-1}+(1+a\Delta t)w^{n+1}_{i_0}-b\Delta tw^{n+1}_{i_0+1}\\
		&=&\left(1-\Delta t\sum_{j=K_l}^{K_r}\nu_j\right)w^{n}_{i_0}+\Delta t\sum_{j=K_l}^{K_r}w^{n}_{i_0+j}\nu_j\geq 0,
	\end{eqnarray*}
	which is a contradiction. Therefore, $\inf\limits_{i\in \Z}w^{n+1}_{i}\geq 0$, and hence $w^{n+1}\geq 0$. 
\end{proof}
\subsection{ Convergence}
As  proved above, our scheme \eqref{fds} is locally consistent, stable, monotone and verifies the discrete comparison principle. In the usual approach to the convergence of finite difference schemes for PDE's, consistency and stability ensure convergence under regularity assumptions on the solution. These conditions are not sufficient here because the solution may not be smooth, and higher order derivatives may not exist. This is where the notion of viscosity solutions are introduced. In the second order parabolic/elliptic PDEs Barles and Souganidis \cite{gb} showed that for elliptic (or parabolic) PDEs, any locally consistent, stable and monotone finite difference scheme converge uniformly, on each compact subset $[0,T]\times\R$, to the unique continuous viscosity solution. Rama Cont and Ekaterina Voltchkova in \cite{eka} showed that the solution of a numerical scheme converges uniformly on each compact subset of $[0,T]\times\R$ to the unique viscosity solution even when the subsolution and the supersolution constructed using a numerical scheme may not have uniform continuity properties. The PIDE studied in this paper relies on the same assumptions as in \cite{eka}, except that here, we are in the case of a finite activity measure since the sizes of the jumps is log-normal. Then we used the same technics to showed the convergence of the explicit-implicit scheme \eqref{fds} to the viscosity solution of problem \eqref{pide}. 
\begin{proposition}(\textbf{Convergence of the explicit-implicit scheme})\\
Let $H$ be a bounded piecewise continuous initial condition, then solution $u^{(\Delta t,\Delta x)}$ of the numerical scheme converges uniformly on each compact subset of $[0,T]\times\R$ to the solution $u$ of continuous problem \eqref{pide}.
\end{proposition}
\begin{proof}
	Define
\begin{equation}\label{supsub}
\begin{cases}
\underline{u}(\tau,x)=\liminf\limits_{(\Delta t,\Delta x)\rightarrow (0,0) \!(t,y)\rightarrow(\tau,x)}u^{(\Delta t,\Delta x)}(t,y)\;\; and\\
\overline{u}(\tau,x)=\limsup\limits_{(\Delta t,\Delta x)\rightarrow(0,0) \! (t,y)\rightarrow(\tau,x)}u^{(\Delta t,\Delta x)}(t,y).
\end{cases}
\end{equation}
The aim of this proof consists to show the following equalities $\underline{u}(\tau,x)=\overline{u}(\tau,x)=u(\tau,x)$. Before showing this equalities some preparatory results are needed.\\
We start by giving an equivalent expression for (\ref{fds}).
\begin{eqnarray}\label{rfds}
u(\tau_n,x_i)&=&F[u(\tau_n-\Delta t,.)](x_i), \, n=1,...,M, \, i\in{0,...,N},\nonumber\\
u(0,x_i)&=&H(S_0e^{x_i}), \,i\in{0,...,N},\\
u(\tau_n,x_i))&=&0, \, n=0,...,M, \, \notin{0,...,N}\nonumber.
\end{eqnarray}
One can define super and subsolution of \ref{rfds} by the following definition
\begin{definition}
A function w is a supersolution of \ref{rfds} if
\begin{eqnarray}\label{supfds}
w(\tau_n,x_i)&\geq&F[w(\tau_n-\Delta t,.)](x_i), \, n=1,...,M, \, i\in{0,...,N},\nonumber\\
w(0,x_i)&\geq&H(S_0e^{x_i}), \,i\in{0,...,N},\\
w(\tau_n,x_i))&\geq&0, \, n=0,...,M, \, \notin{0,...,N}\nonumber.
\end{eqnarray}
A function z is a subsolution of \ref{rfds} if
\begin{eqnarray}\label{subfds}
z(\tau_n,x_i)&\leq&F[z(\tau_n-\Delta t,.)](x_i), \, n=1,...,M, \, i\in{0,...,N},\nonumber\\
z(0,x_i)&\leq&H(S_0e^{x_i}), \,i\in{0,...,N},\\
z(\tau_n,x_i))&\leq&0, \, n=0,...,M, \, i\notin{0,...,N}\nonumber.
\end{eqnarray} 
\end{definition}

To Avoid the problem of uniform continuity and smoothness which may not hold for $\underline{u}$ and $\overline{u}$ define in (\ref{supsub}), it is convenient to consider smooth super and subsolutions of \ref{pide} and super and subsolutions of \ref{rfds}, and to derive the link with $\underline{u}$ and $\overline{u}$. The following results extends  
the  comparison principle to the super and subsolutions. 
\begin{lemma}\label{lemma1}
For any supersolution $w$ and subsolution $z$ of \ref{rfds} we have
$z\leq u\leq w$.
\end{lemma}
\begin{proof} 
For ($i\notin{0,...,N}$) or ($n=0$ and $i \in{0,...,N}$) the above inequalities are satisfied by definition. For $n=1,...,M$, 
\, $i\in{0,...,N}$ from monotonicity of the scheme we have
\begin{eqnarray*}
	z(\tau_n,x_i)&\leq& F[z(\tau_n-\Delta t,.)](x_i)\leq F[u(\tau_n-\Delta t,.)](x_i)=u(\tau_n,x_i)\\
	&=& F[u(\tau_n-\Delta t,.)](x_i)\leq  F[w(\tau_n-\Delta t,.)](x_i)\leq w(\tau_n,x_i).
\end{eqnarray*}
\end{proof}
\begin{lemma}\label{lemma2}
Let $w$ and $z$ be a smooth supersolution and subsolution of \ref{pide} respectively. Then for all $\epsilon$, there exists $\Delta>0$ such that\\
$\forall \Delta t, \Delta x,\leq\Delta $, \, $\forall n\geq 0$, \, $\forall i\in \mathbb{Z},$\, $z(\tau_n,x_i)-\epsilon<u(\tau_n,x_i)<w(\tau_n,x_i)+\epsilon$
\end{lemma}
\begin{proof}
Choose q such that $0<q(T+1)<\epsilon$ and let $\tilde{w}(\tau,x)=w(\tau,x)+q(1+\tau)$, notice that a constant function is always a solution. In fact one can see from the definition that the scheme is linear.\\
If  $i\notin{0,...,N}$, we have
\begin{eqnarray}\label{l1}
 \tilde{w}(\tau_n,x_i)=w(\tau_n,x_i)+q(\tau+1)\geq q\geq 0.
\end{eqnarray}
If  $n=0$ and $i \in{0,...,N}$, 
\begin{eqnarray}
\tilde{w}(0,x_i)=w(0,x_i)+q\geq H(S_0e^{x_i}).
\end{eqnarray}
If $n\geq 1$, \, $i\in{0,...,N}$  from the consistency of the scheme we obtain 
\begin{eqnarray}\label{l2}
\dfrac{\tilde{w}(\tau_n,x_i)-\tilde{w}(\tau_n-\Delta t,x_i)}{(\Delta t)}-
L_D\tilde{w}(\tau_n-\Delta t,x_i)-L_J\tilde{w}(\tau_n-\Delta t,x_i)&=&
\dfrac{w(\tau_n,x_i)-w(\tau_n-\Delta t,x_i)}{(\Delta t)}\nonumber\\-
L_Dw(\tau_n-\Delta t,x_i)-L_Jw(\tau_n-\Delta t,x_i)+q>0\\&\longrightarrow&
\dfrac{\partial w}{\partial \tau}(\tau,x)-(\mathcal{L}_D+\mathcal{L}_J)w(\tau,x)+q\nonumber
\end{eqnarray}
as $\Delta t, \, \Delta x \longrightarrow (0,0)$, \, $(\tau_n,x_i)\longrightarrow(\tau,x)$, uniformly on $(0,T[\times O$. Therefore for any sufficiently small $\Delta>0$, for all $\Delta t, \, \Delta x \leq \Delta$, we have
\begin{eqnarray}\label{l3}
	\tilde{w}(\tau_n,x_i)\geq F[\tilde{w}(\tau_n-\Delta t,.)](x_i),\;\; \forall n\leq 1, \, \forall i\in{0,...,N}).
\end{eqnarray}  	
Combining \ref{l1}, \ref{l2} and \ref{l3}, show that function $\tilde{w}$ is supersolution of \ref{rfds}. Indeed, Lemma \ref{lemma1} implies that
\begin{eqnarray*}
u(\tau_n,x_i)\leq \tilde{w}(\tau_n,x_i)+q(1+T)<w(\tau_n,x_i)+\epsilon,      \;\;\forall n\geq 0, \; \forall i\in\mathbb{Z},
\end{eqnarray*}
which is the desired property. The lower bound $z(\tau_n,x_i)-\epsilon $ can be proved in the same manner and then completes the proof.
\end{proof}
Following Lemma \ref{l1} and Lemma \ref{l2}, we have the following main Lemma
\begin{lemma}\label{lemma3}
Let $\underline{u}$ and $\overline{u}$ be the function define by \ref{supsub}. For any smooth supersolution $w(\tau,x)$ and any subsolution $z(\tau,x)$ of the problem \ref{pide}, we have for $(\tau,x)\in[0,T]\times O$,
\begin{eqnarray}\label{rl3}
z(\tau,x)\leq \underline{u}(\tau,x)\leq\overline{u}(\tau,x)\leq w(\tau,x).
\end{eqnarray}
\end{lemma}
 \begin{proof}
 By the definition of upper and lower limits, Lemma \ref{lemma2} implies desired property.
 \end{proof}
 After giving some properties needed we can start the proof of convergence (i.e. showed that $\underline{u}=\overline{u}=u$). If $\overline{H}, \,\underline{H}$ are smoothness functions on $\mathbb{R}$ such that $\underline{H}\leq H\leq \overline{H}$, then $w(\tau,x)=\E[\overline{H}(Y^{x}_{\tau})]$ and $z(\tau,x)=\E[\underline{H}(Y^{x}_{\tau})]$ are respectively a supersolution and a subsolution of the Cauchy problem \ref{pide}. From Lemma \ref{lemma3} we obtain \ref{rl3}. Notice that If $w(\tau,x)-u(\tau,x), \; u(\tau,x)-z(\tau,x)$ could be made small this would imply that $\lim\limits_{(\Delta t,\Delta x)\rightarrow(0,0) \!(\tau_n,x_i)\rightarrow(\tau,x)}u^{(\Delta t,\Delta x)}(\tau_n,x_i)=u(\tau,x)$. Indeed, it remains to construct appropriate smooth approximations $\overline{H}$ and $\underline{H}$.\\
 Let $\zeta_1,....,\zeta_I$ be the discontinuity points of $H$. We suppose that the jumps of $H$ are bounded by $\delta$. 
Given $\epsilon>0$ and $\overline{H}, \,\underline{H}$ smooth functions that satisfied the following relations
\begin{eqnarray*}
\underline{H}(x)\leq H\leq\overline{H}(x)&& \;\;\;\;\forall x\in\mathbb{R},\\
\mid\underline{H}(x)-\overline{H}(x)\mid&\leq&\delta\;\;\;   \forall x\in\bigcup\limits_{j=1}^{I}(\zeta_j-\epsilon,\zeta_j+\epsilon),\\
\mid\underline{H}(x)-\overline{H}(x)\mid&\leq&\epsilon \;\;\;  \forall x\notin\bigcup\limits_{j=1}^{I}(\zeta_j-\epsilon,\zeta_j+\epsilon).\\
\end{eqnarray*}
We have
\begin{eqnarray}
w(\tau,x)-z(\tau,x)&=&\E[\overline{H}(Y_{\tau}^{x})-\underline{H}(Y_{\tau}^{x})]\nonumber\\
&\leq& \delta\Q(Y_{\tau}^{x}\in\bigcup\limits_{j=1}^{I}(\zeta_j-\epsilon,\zeta_j+\epsilon))+\epsilon\Q(Y_{\tau}^{x}\notin\bigcup\limits_{j=1}^{I}(\zeta_j-\epsilon,\zeta_j+\epsilon))\\
&\leq&\delta\Q(Y_{\tau}^{x}\in\bigcup\limits_{j=1}^{I}(\zeta_j-\epsilon,\zeta_j+\epsilon))+\epsilon.
\end{eqnarray}
Noting that $\bigcap\limits_{\epsilon>0}\{ Y_{\tau}^{x}\in\bigcup\limits_{j=1}^{I}(\zeta_j-\epsilon,\zeta_j+\epsilon)\}=\{Y_{\tau}^{x}\in\{\zeta_1,....,\zeta_I\}\}$. Since $Y_{\tau}^{x}$
has an absolutely continuous distribution, so we have $\Q(\{Y_{\tau}^{x}\in\{\zeta_1,....,\zeta_I\}\})=0$. Consequently
$\Q(Y_{\tau}^{x}\in\bigcup\limits_{j=1}^{I}(\zeta_j-\epsilon,\zeta_j+\epsilon))\longrightarrow0$ as $\epsilon\longrightarrow0$.\\ Therefore $w(\tau,x)-z(\tau,x)\longrightarrow0$ as $\epsilon\longrightarrow0$ and the inequalities $z(\tau,x)\leq u(\tau,x)\leq w(\tau,x) $ together with Lemma \ref{lemma3} implies desired result which completes the proof.  
\end{proof}
\begin{remark}
For $\tau=0$ the scheme does not converge to the initial condition at the discontinuous points of $H$. This is due to the fact that
$\Q(Y_{\tau}^{x}\in\bigcup\limits_{j=1}^{I}(\zeta_j-\epsilon,\zeta_j+\epsilon))\longrightarrow\Q(S_0e^{x}\in\{\zeta_1,....,\zeta_I\})=\mathds{1}_{\{x\in\{\ln(\frac{\zeta_1}{S_0}),...,\ln(\frac{\zeta_I}{S_0})\}\}}$.
However, this has no practical interest since it is not important to compute the solution numerically at $\tau=0$.
\end{remark}

\section{Numerical Results}
In this section we discuss to the details of the implementation of our schemas and present numerical results and some interpretation.

 Before started simulation parameters scheme are take as follow
\begin{table}[!ht]
	\caption{scheme parameters }
	\begin{tabular}{ccccc}
		\hline\noalign{\smallskip}
		$T$ &$M$&$N$&$A_l$ &$A_r$\\
		\noalign{\smallskip}\hline\noalign{\smallskip}
		$1$&$100$&$175$&$-0.096$&$0.079$ \\
		\noalign{\smallskip}\hline
	\end{tabular}
\end{table}
 the parameters used to implement the following results are taking independent to time as follow. 
 
\begin{table}[!ht]
	\caption{Model parameters }
	\label{tab:1}     
	\begin{tabular}{ccccc}
		\hline\noalign{\smallskip}
		figure & model &$r$&Strike&Product\\
		\noalign{\smallskip}\hline\noalign{\smallskip}
		\ref{fig:1}, \ref{fig:2}, \ref{fig:3}, \ref{fig:4} & $\alpha=0.015$  $\beta=0.4$  $\sigma_J=0.5$  $\ell=1.5$  $\sigma=0.5$  $S_0=50$&$0.04$&$K=45$&Call \\
		\noalign{\smallskip}\hline
	\end{tabular}
\end{table}
\begin{figure}[!ht]
\includegraphics[width=7cm]{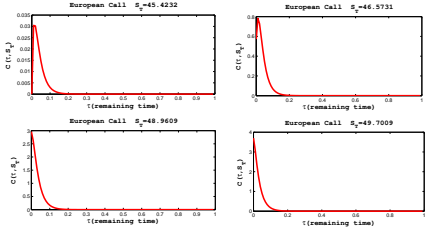}
\caption{Call price for four diffrents values of spot price at maturity versus remaining time to maturity}
\label{fig:1}      
\end{figure}
\begin{figure}[!ht]
\includegraphics[width=7cm]{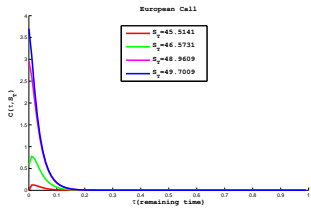}
\caption{Comparison of call values for four diffrents values of spot price at maturity versus remaining time to maturity }
\label{fig:2}       
\end{figure}
Figure \ref{fig:1} plots call option for different spot price of the underlying at maturity versus remaining time. Note that, as remaining time increase the mean-reverting and price-cap effect in all four case cause the call prices converge to zero under the set parameters gived above.
Figure \ref{fig:2} give comparison between illustrative curve of call for four different electricity spot price at maturity. This two figures illustrate that for different initials conditions call option prices converge verse to the same value. From this we can thus, say that different properties shown theoretically are effective under the established conditions.
\begin{figure}[!ht]
	\includegraphics[width=7cm]{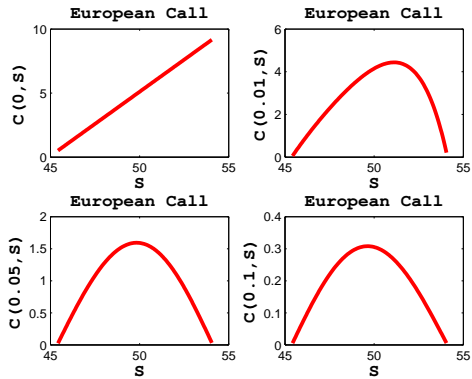}
	\caption{Call price for four different remaining time to maturity versus spot price}
	\label{fig:3}       
\end{figure}
\begin{center}
\begin{figure}[!ht]
	\includegraphics[width=7cm]{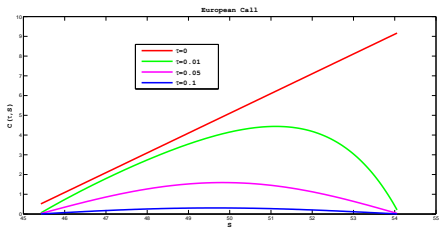}
	\caption{Comparison of call price for four different versus spot price }
	\label{fig:4}       
\end{figure}
Figure \ref{fig:3} and \ref{fig:4} illustrate a reality enough close to those in the classical financials markets. In the sens that in financial market the values of call option before the maturity evolved in the form of a curve which towards to the payoff line progressively and as we approach maturity.   
\end{center}
\begin{figure}[!ht]
	\includegraphics[width=7cm]{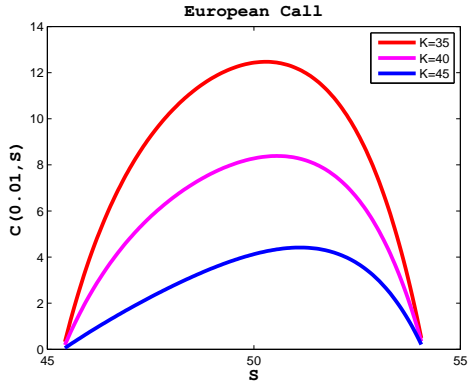}
	\caption{Call price for three diffrents values of strike versus spot price S, the other parameters is unchanged as in table \ref{tab:1}}
	\label{fig:5}      
\end{figure}
\begin{figure}[!ht]
	\includegraphics[width=7cm]{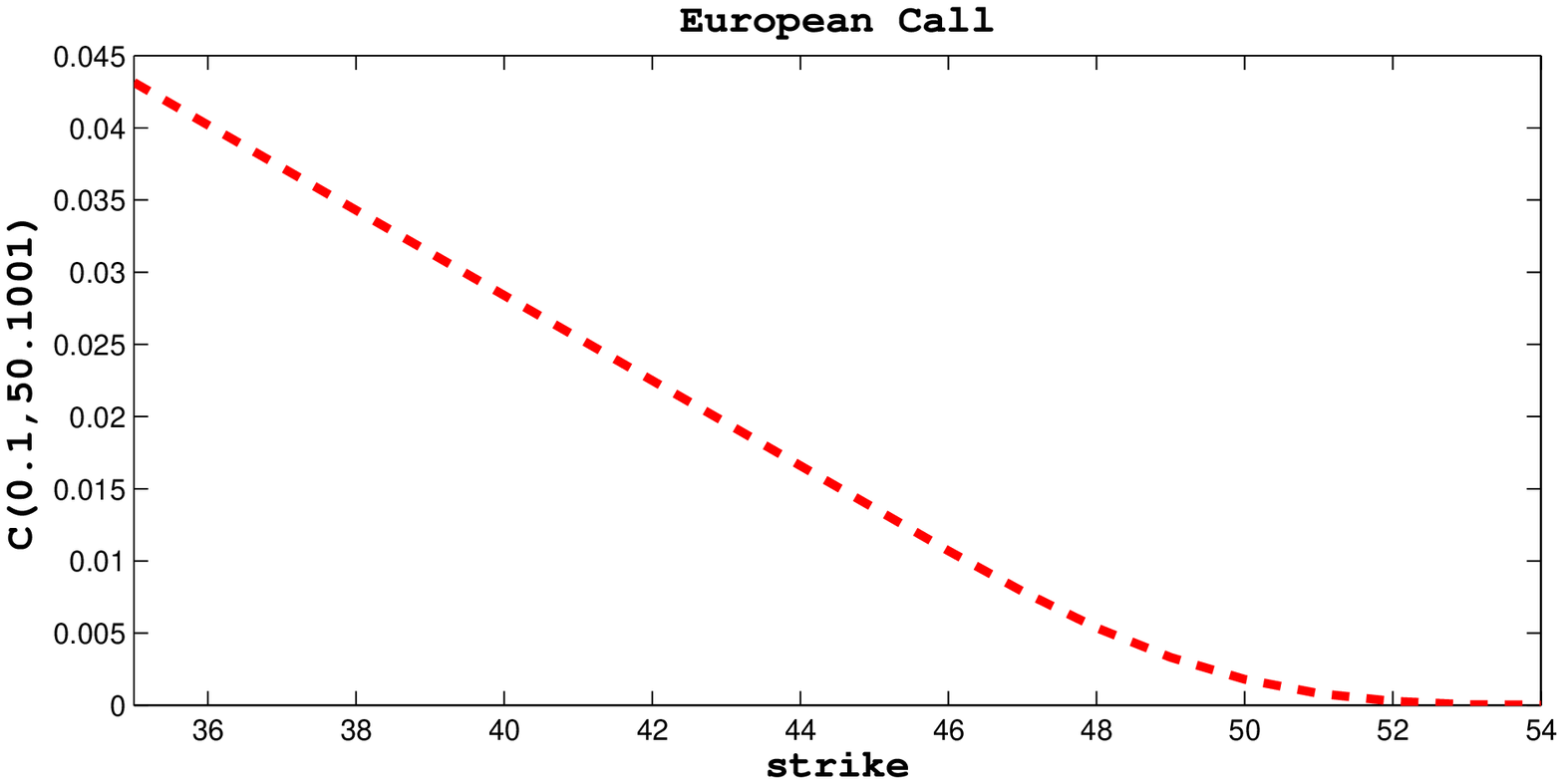}
	\caption{ Call price versus strike price K, the other parameters in unchange as in table \ref{tab:1}}
	\label{fig:6}      
\end{figure}
\begin{figure}[!ht]
	\includegraphics[width=7cm]{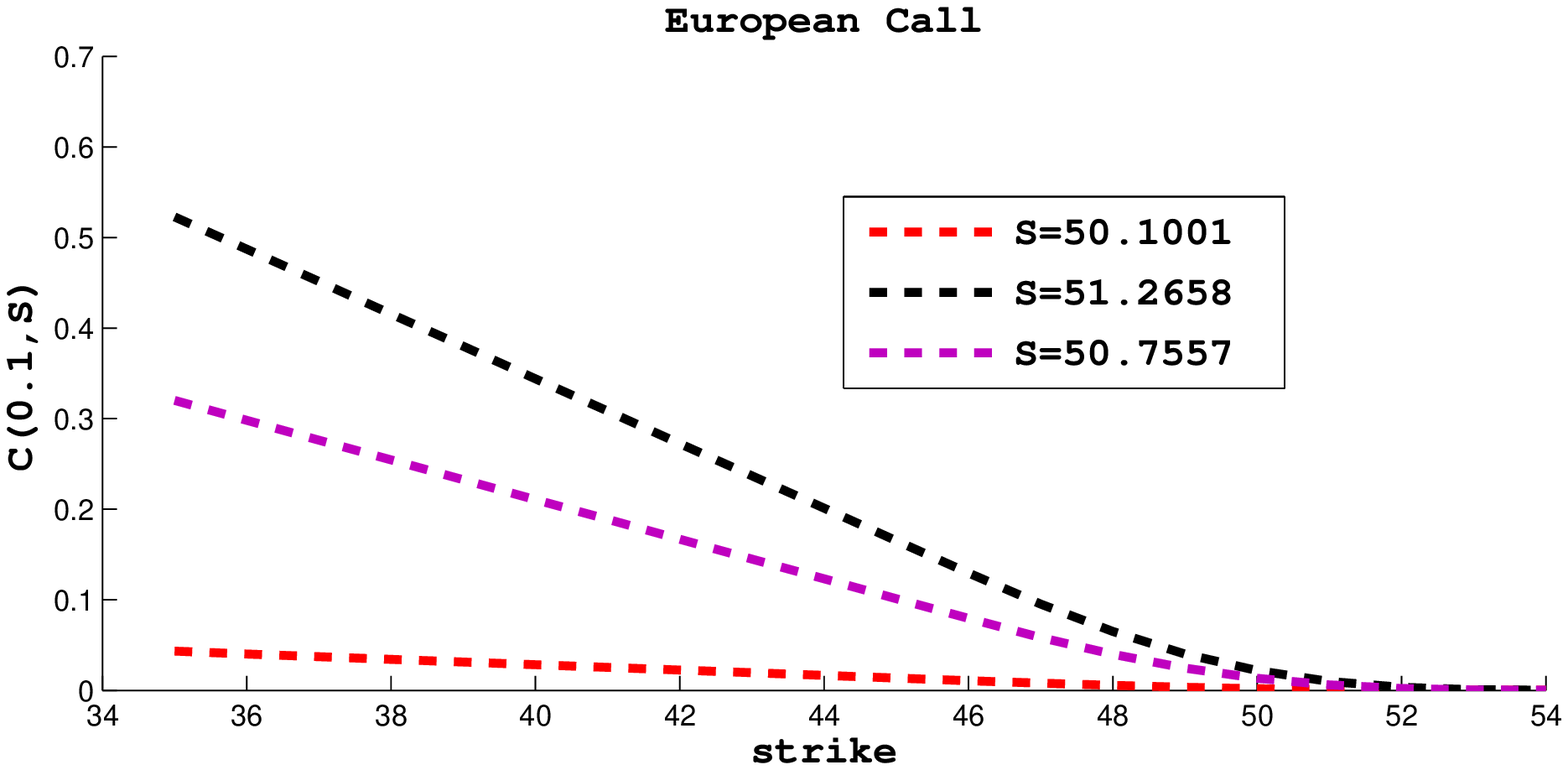}
	\caption{ Call price versus strike price K, the other parameters in unchange as in table \ref{tab:1}}
	\label{fig:7}      
\end{figure}
From figures \ref{fig:5}, \ref{fig:6} and \ref{fig:7} one can observed that the values of call decrease when strike price increase. This behaviour of the call values is from a risk management point of view what it is wished.    
\begin{figure}[!ht]
	\includegraphics[width=7cm]{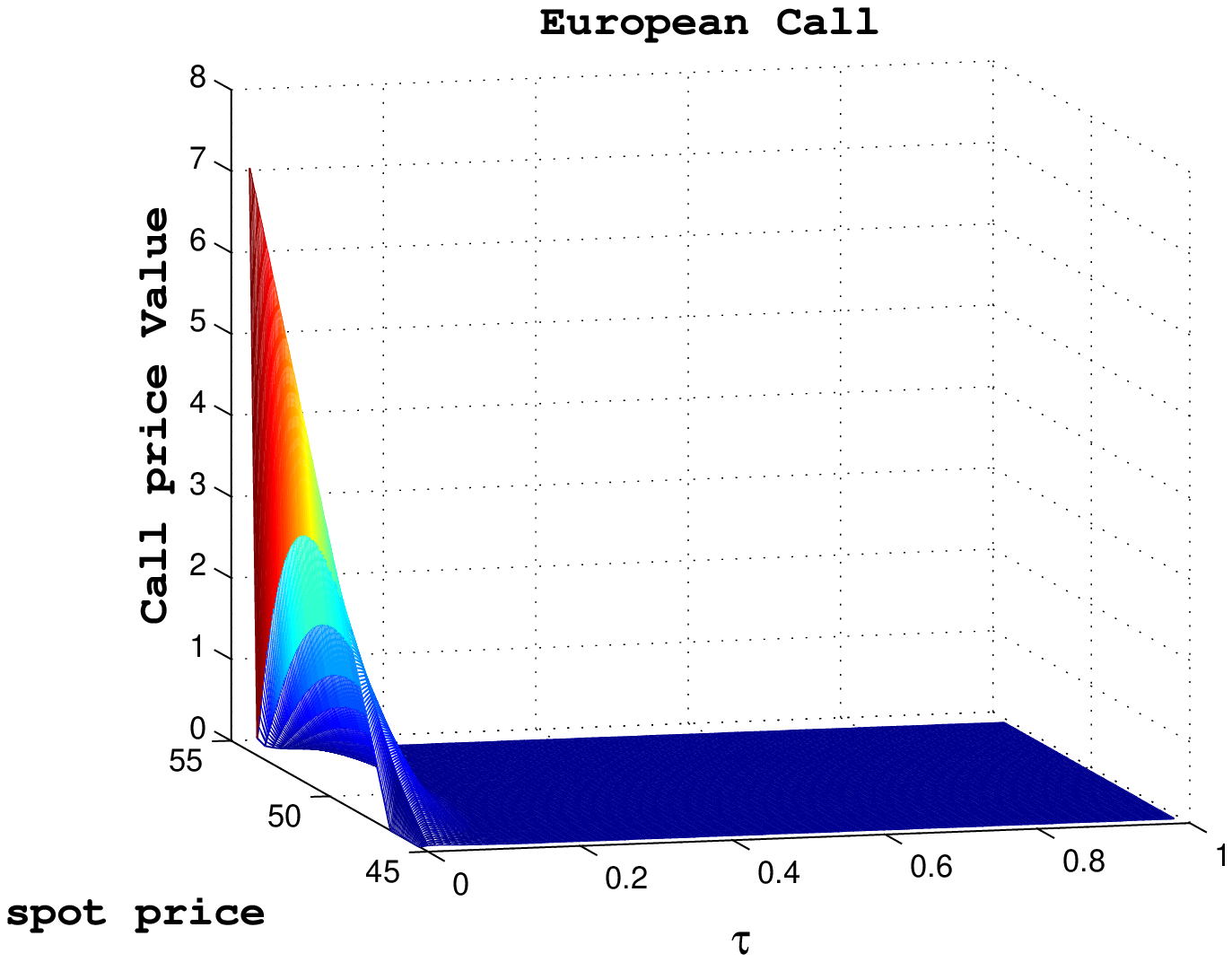}
	\caption{ Call price versus remaining time and spot price $\sigma_J=0.5$, $K=45$, the other parameters is unchanged as in table \ref{tab:1} }
	\label{fig:8}      
\end{figure}
\begin{figure}[!ht]
	\includegraphics[width=7cm]{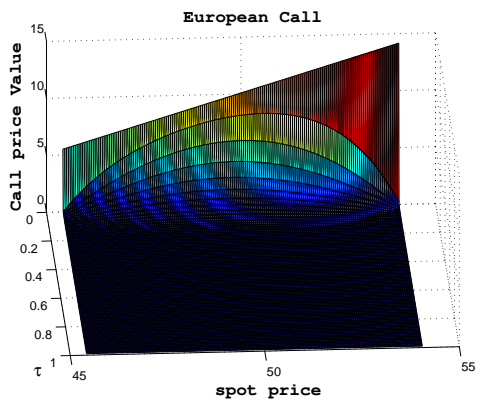}
	\caption{ Call price versus remaining time and spot price  $\sigma_J=2.5$ and $K=40$, the other parameters is unchanged as in table \ref{tab:1}}
	\label{fig:9}      
\end{figure}
Analysis plots of figures \ref{fig:8} and \ref{fig:9} which illustrate call option price as a function of remaining time and spot price one can said that for a large remaining time jumps effects are not perceptible and then can effect call. Whereas a small remaining time to maturity call prices increase suddenly which express the effect of jump. We must therefore say that the jump term which allows to take into account certain reality of the electricity market is not inconsiderable since that it impact on the call prices are quite noticeable.
\newpage
\section{Discussion and Conclusion}
The characterization of european options prices in terms of the classical solution, or, in general, in the terms of the viscosity solution of a PIDE allows the use of numerical methods to obtain efficient approximations of option values. This has been a centre of research in recent times in the case of exponential L\'{e}vy models with finite arrival or infinite arrival rate of jumps. Some authors use the finite difference method to approximate the PIDE solution (see \cite{alvarez}), while others like \cite{ekafinite} approximate viscosity solutions in the case of nonsmoothness of option prices. In both cases, success (in terms of efficient approximation) has been obtained. In this paper we used their approach to evaluate European call option when the underlying is electricity. The motivation behind our approach arose from the fact that the electricity prices model presented here, by hypothesis, possesses most of the properties (as in their case) of an exponential L\'{e}vy model, and the Markov process property. We focus on the pricing of call option because put option can be deduced from the put-call using parity formula. In the mathematical point of view, numerical results have confirmed the established theoretical results. In the finance point of view, numerical results present an interpretation which was coherent with some realities in the electricity market, when it is regulated under price cap. 

A limitation of this work is that it can be adapted only to option pricing with short maturity in a regulated market. In future work, we plan to study option pricing on future option.            
\section*{Acknowledgements}
We thank the African Center of Excellence in Technologies, Information and Communication ( CETIC ) which placed at our disposal its infrastructures. This helped us to improve conditions of work.

\end{document}